\documentclass{amsart}       

\newtheorem{defn}{Definition}[section]
\newtheorem{remark}[defn]{Remark}

\newtheorem{lemma}[defn]{Lemma}
\newtheorem{theorem}[defn]{Theorem}
\newtheorem{corollary}[defn]{Corollary}

\usepackage{fourier}
\usepackage{enumerate}
\usepackage{graphicx}

\newcommand{\algo}{\mathcal{A}}
\newcommand{\M}{\mathcal{M}}
\newcommand{\F}{\mathcal{F}}

\newcommand{\cvginprob}{\overset{p}\to}

\usepackage[colorlinks=true,breaklinks=true,bookmarks=true,urlcolor=blue,
     citecolor=blue,linkcolor=blue,bookmarksopen=false,draft=false]{hyperref}

\title{Asymptotic welfare performance of Boston assignment algorithms}
\date{\today}
\author{Geoffrey Pritchard}
\author{Mark C. Wilson}
\begin{document}

\begin{abstract}
We make a detailed analysis of three key algorithms (Serial Dictatorship and the naive and adaptive variants of the Boston algorithm) for the housing allocation problem, under the assumption that agent preferences are chosen iid uniformly from linear orders on the items.
We compute limiting distributions (with respect to some common utility functions) as $n\to \infty$ of both the utilitarian  welfare  and the order bias. To do this, we compute limiting distributions of the outcomes for an arbitrary agent whose initial relative position in the tiebreak order is $\theta\in[0,1]$, as a function of $\theta$. We expect that these fundamental results on the stochastic processes underlying these mechanisms will have wider applicability in future. Overall our results show that the differences in utilitarian welfare performance of the three algorithms are fairly small, but the differences in order bias are much greater. Also, Naive Boston beats Adaptive Boston, which beats Serial Dictatorship, on both welfare and order bias.
\end{abstract}


\maketitle

%


\section{Introduction}
\label{s:intro}

Algorithms for allocation of indivisible goods are widely applicable and have been heavily studied. There are many variations on the problem, for example one-sided matching or housing allocation (each agent gets a unique item), school choice (each student gets a single school seat, and schools have limited preferences over students), and multi-unit assignment (for example each student is allocated a seat in each of several classes). One can also vary the type of preferences for agents over items, but here we focus on the most commonly studied case, of complete strict preferences. We focus on the housing allocation problem \cite{HyZe1979}, whose relative simplicity allows for more detailed analysis. 

\subsection{Our contribution}
\label{ss:contrib}

We make a detailed analysis of three prominent algorithms (Serial Dictatorship and the naive and adaptive variants of the Boston algorithm) for the housing allocation problem, under the standard assumption that agent preferences are independently chosen uniformly from linear orders on the items (often called the Impartial Culture distribution), and the further assumption that agents express truthful preferences. 

We compute limiting distributions as $n\to \infty$ of both the utilitarian  welfare (with respect to some common utility functions) and the order bias (a recently introduced \cite{FrPW2021} fairness concept). In order to do this, we compute limiting distributions of the outcomes for an arbitrary agent whose initial relative position in the tiebreak order is $\theta\in[0,1]$, as a function of $\theta$. We expect that these fundamental results on the stochastic processes underlying these mechanisms will have wider applicability in future. While the results for Serial Dictatorship are easy to derive, the Boston mechanisms require substantial work. 

To our knowledge, no precise results of this type on average-case welfare performance of allocation algorithms  have been published. In Section~\ref{s:conclude} we discuss the limitations and implications of our results, situate our work in the literature on welfare of allocation mechanisms, and point out opportunities for future work.

We first derive the basic limiting results for exit time and rank of the item attained, for Naive Boston, Adaptive Boston and Serial Dictatorship in Sections~\ref{s:naive_asymptotics}, \ref{s:adaptive_asymptotics} and \ref{s:serial_dictatorship} respectively. 
Each section first deals with average-case results for an arbitrary initial segment of agents in the choosing order, and then with the fate of an individual agent at an arbitrary position. The core technical results are found in Theorems~\ref{thm:naive_asym_agent_numbers}, \ref{thm:naive_individual_asymptotics}, \ref{thm:adaptive_asym_agent_numbers}, \ref{thm:adaptive_asymptotics},  \ref{thm:adaptive_individual_asymptotics} 
and their corollaries.
We apply the basic results to utilitarian welfare in Section~\ref{s:welfare} and order bias in Section~\ref{s:order_bias}, and discuss the implications, relation to previous work, and ideas for possible future work in Section~\ref{s:conclude}.

The results for Serial Dictatorship are straightforwardly derived, but the other algorithms require nontrivial analysis. Of those, Naive Boston is much easier, because the nature of the algorithm means that the exit time of an agent immediately yields the preference rank of the item obtained by the agent. However in Adaptive Boston this link is much less direct and this necessitates substantial extra technical work.

\section{Preliminaries}
\label{s:preliminaries}

We define the mechanisms Naive Boston, Adaptive Boston and Serial Dictatorship, and show how to model the assignments they give via stochastic processes.

\subsection{The mechanisms}
\label{s:algo}

We assume throughout that we have $n$ agents and $n$ items, where each agent has a complete strict preference ordering of items. Each mechanism allows for strategic misrepresentation of preferences by agents, but we assume sincere behavior here for this baseline analysis. We are therefore studying the underlying preference aggregation algorithms. These can be described as centralized procedures that take an entire preference profile and output a matching of agents to items, but are more easily and commonly interpreted dynamically as explained below.

Probably the most famous mechanism for housing allocation is \emph{Serial Dictatorship} (SD). In a common implementation, agents choose according to the exogenous order $\rho$,  each agent in turn choosing the item he most prefers among those still available.

The Boston algorithms in the housing allocation setting are as follows.
\emph{Naive Boston} (NB) proceeds in rounds: in each round, some of the agents and items will be permanently matched, and the rest will be relegated to the following round. At round $r$ ($r=1,2\ldots$), each remaining
unmatched agent bids for his $r$th choice among the items, and will be matched to that item if it is still available.
If more than one agent chooses an item, then the order $\rho$ is used as a tiebreaker.

\emph{Adaptive Boston} (AB) \cite{MeSe2014} differs from Naive Boston in the set of items available at each round. In each round of this algorithm, all remaining agents submit a bid for their most-preferred item among those still
available at the start of the round, rather than for their most-preferred item among those for which they have not yet bid. The Adaptive Boston algorithm takes fewer rounds to finish than the naive version, because agents do not waste time bidding for their $r$th choice in round $r$ if it has already been assigned to someone else in a previous round. This means that the algorithm runs more quickly, but agents, especially those late in the choosing order, are more likely to have to settle for lower-ranked items. Note that both Naive and Adaptive Boston behave exactly the same in the first round, but differently thereafter.

\subsection{Important stochastic processes in the IC model}
\label{s:stoch}

Under the Impartial Culture assumption, it is convenient to imagine the agents developing their preference
orders as the algorithm proceeds, rather than in advance. This allows the evolution of the assignments for the Boston algorithms
to be described by the following stochastic processes (for SD the analysis is easier).

In the first round, the naive and adaptive Boston processes proceed identically: each agent randomly
chooses one of the $n$ items, independently of other agents and with uniform probabilities 
$\frac{1}{n}$, as his most preferred item for which to bid. Each item that is so chosen is
assigned to the first (in the sense of the agent order $\rho$) agent who bid for it;
items not chosen by any agent are relegated, along with the unsuccessful agents, to the next round.
In the $r$th round ($r\geq2$), the naive algorithm causes each remaining agent to randomly choose
his $r$th most-preferred item, independently of other agents and of his own previous choices, uniformly
from the $n-r+1$ items for which he has not previously bid. (Note that included among these are all the items
still available in the current round.) Each item so chosen is assigned to the first agent who chose it;
other items and unsuccessful agents are relegated to the next round. The adaptive Boston method is similar,
except that agents may choose only from the items still available at the start of the round.
This can be achieved by having each remaining agent choose his next most-preferred item by repeated
sampling without replacement from the set of items he has not yet considered, until one of the items sampled
is among those still available at this round.

An essential feature of these bidding processes is captured in the following two results.

\begin{lemma}
\label{lem:A}
Suppose we have $m$ items ($m\geq2$) and a sequence of agents (Agent 1, Agent 2, $\ldots$) who each randomly
(independently and uniformly) choose an item. Let $A\subseteq{\mathbb N}$ be a subset of the agents, and
$C_A$ be the number of items first chosen by a member of $A$. (Equivalently, $C_A$ is the number of members
of $A$ who choose an item that no previous agent has chosen.) Then
$$ \hbox{Var}(C_A) \;\leq\; E[C_A] \;=\; \sum_{a\in A} \left(1-\frac1m\right)^{a-1}
  $$
\end{lemma}

\begin{lemma}
\label{lem:Aprime}
Suppose we have the situation of Lemma \ref{lem:A}, with the further stipulation that $\ell$ of the $m$ items
are blue. Let $C_A$ be the number of blue items first chosen by a member of $A$ (equivalently, the number of
members of $A$ who choose a blue item that no previous agent has chosen.) Then
$$ \hbox{Var}(C_A) \;\leq\; E[C_A] \;=\; \frac{\ell}{m} \sum_{a\in A} \left(1-\frac1m\right)^{a-1}
  $$
\end{lemma}

\begin{remark}
Lemmas \ref{lem:A} and \ref{lem:Aprime} are applicable to the adaptive and naive Boston mechanisms, respectively.
The blue items in Lemma \ref{lem:Aprime} correspond to those still available at the start of the round.
In the actual naive Boston algorithm, the set of unavailable items that an agent may still bid for will typically
be different for different agents, but the number of them ($m-\ell$) is the same for all agents, which is
all that matters for our purposes.
\end{remark}

\begin{proof}{Proof of Lemmas \ref{lem:A} and \ref{lem:Aprime}.}
Lemma \ref{lem:A} is simply the special case of Lemma \ref{lem:Aprime} with $\ell=m$, so the following
direct proof of Lemma \ref{lem:Aprime} suffices for both.
\newcommand{\rrm}{\left(1-\frac1m\right)}
Let $F_i$ denote the agent who is first to choose item $i$, and $X_{ia}$ the indicator of the event
$\{F_i=a\}$. That is, $X_{ia}=1$ if and only if $F_i=a$.
We have $P(F_i=a)=\frac1m\rrm^{a-1}$: agent $a$ must choose $i$, while all previous agents choose items
other than $i$.
Let $B$ be the set of blue items.
Then $C_A=\sum_{i\in B} \sum_{a\in A} X_{ia}$, so
$$ E[C_A] \;=\; \sum_{i\in B} \sum_{a\in A} P(F_i=a)
    \;=\; \sum_{i\in B} \sum_{a\in A} \frac1m \rrm^{a-1} 
    \;=\; \frac{\ell}{m} \sum_{a\in A} \rrm^{a-1},
  $$
as claimed. Also,
\begin{equation}
\label{eq:ecasq}
E\left[C_A^2\right] \;=\; \sum_{i\in B} \sum_{j\in B} \sum_{a\in A} \sum_{b\in A} E[X_{ia}X_{jb}] .
\end{equation}
For $a\neq b$ these summands are identical for all $i\neq j$ (and zero for $i=j$);
for $a=b$, they are identical for all $i=j$ (and zero for $i\neq j$). Thus
(\ref{eq:ecasq}) reduces to
\begin{equation}
\label{eq:ecasqq}
E\left[C_A^2\right] \;=\; \ell(\ell-1) \sum_{a,b\in A; a\neq b} E[X_{1,a}X_{2,b}]
     \;+\; \ell \sum_{a\in A} E[X_{1,a}] .
\end{equation}
The second term of (\ref{eq:ecasqq}) is $E[C_A]$ again.
For $a<b$ and $i\neq j$ we have
$$ E[X_{ia}X_{jb}] \;=\; P(F_i=a\hbox{ and }F_j=b) = \left(1-\frac2m\right)^{a-1} \frac1m \rrm^{b-a-1} \frac1m
  $$
(Agents prior to $a$ must choose neither $i$ nor $j$, $a$ must choose $i$, agents between $a$ and $b$
must choose items other than $j$, and $b$ must choose $j$.) Since $1-\frac2m < \rrm^2$, this gives
\begin{equation}
\label{eq:xiaxjb}    
 E[X_{ia}X_{jb}] \;\leq\; \frac{1}{m^2} \rrm^{a+b-3} .
\end{equation}
As this last expression is symmetric in $a$ and $b$, (\ref{eq:xiaxjb}) also holds for $a>b$.
Hence,
$$ E\left[C_A^2\right] \;\leq\; E[C_A] \;+\;
  \frac{\ell(\ell-1)}{m^2} \sum_{a,b\in A; a\neq b} \rrm^{a+b-3} .
  $$
We have
$\frac{\ell(\ell-1)}{m^2} = \frac{\ell^2}{m^2}\left(1-\frac1\ell\right)\leq \frac{\ell^2}{m^2}\rrm$,
since $\ell\leq m$. This gives
$$ E\left[C_A^2\right] \;\leq\; E[C_A] \;+\;
  \frac{\ell^2}{m^2} \sum_{a,b\in A; a\neq b} \rrm^{a+b-2} ,
  $$
enabling us to bound the variance as required: $\hbox{Var}(C_A) = E\left[C_A^2\right] - E[C_A]^2$ and so
\begin{eqnarray*}
\hbox{Var}(C_A) - E[C_A] &\leq&
  \frac{\ell^2}{m^2}\sum_{a,b\in A; a\neq b} \rrm^{a+b-2} \;-\; \left(\frac{\ell}{m} \sum_{a\in A} \rrm^{a-1} \right)^2 \\
&=& \frac{\ell^2}{m^2}\left( \sum_{a,b\in A; a\neq b} \rrm^{a+b-2} \;-\; \sum_{a,b\in A} \rrm^{a+b-2} \right)\\
&\leq& 0 .
\end{eqnarray*}
\hfill \qedsymbol
\end{proof}

The bounding of the variance of a random variable by its mean implies a distribution with relatively little
variation about the mean when the mean is large. We put this to good use in the following two results.

\begin{lemma}
\label{lem:C}
Let $(X_n)$ be a sequence of non-negative random variables with $\hbox{Var}(X_n)\leq E[X_n]$ and
$\frac1n E[X_n]\to c$ as $n\to\infty$. Then $\frac1n X_n\cvginprob c$ as $n\to\infty$ (convergence in probability).
\end{lemma}

\begin{lemma}
\label{lem:CC}
Let $(X_n)$ be a sequence of non-negative random variables and $(\F_n)$ a sequence of $\sigma$-fields,
with $\hbox{Var}(X_n|\F_n)\leq E[X_n|\F_n]$ and $\frac1n E[X_n|\F_n]\cvginprob c$ as $n\to\infty$.
Then $\frac1n X_n\cvginprob c$ as $n\to\infty$.
\end{lemma}

\begin{proof}{Proof of Lemmas \ref{lem:C} and \ref{lem:CC}.}
Lemma \ref{lem:C} is just the special case of Lemma \ref{lem:CC} in which all the $\sigma$-fields $\F_n$
are trivial. 
For a proof of Lemma \ref{lem:CC}, it suffices to show that $\frac1n (X_n - E[X_n|\F_n])\cvginprob 0$.
For any $\epsilon>0$ we have by Chebyshev's inequality (\cite{Durrett})
$$
P\left(\Big|X_n - E[X_n|\F_n]\Big| > \epsilon n\Big|\F_n\right)
\;\leq\; (\epsilon n)^{-2} \hbox{Var}(X_n|\F_n)
\;\leq\; (\epsilon n)^{-2} E[X_n|\F_n] 
  $$
Since $n^{-2} E[X_n|\F_n]\cvginprob 0$, it follows that
$P\left(\Big|X_n - E[X_n|\F_n]\Big| > \epsilon n\Big|\F_n\right)\cvginprob0$.
As these conditional probabilities are a bounded (and thus uniformly integrable) sequence, the
convergence is also in ${\mathcal L}_1$ (Theorem 4.6.3 in \cite{Durrett}), and so
$$ P\left(\frac1n\Big|X_n - E[X_n|\F_n]\Big| > \epsilon\right)
\;=\; E\left[P\left(\Big|X_n - E[X_n|\F_n]\Big| > \epsilon n\Big|\F_n\right)\right]
\;\to\; 0 ,
  $$
giving the required convergence in probability.
\hfill \qedsymbol
\end{proof}

The introduction of asymptotics ($n\to\infty$) implies that we are considering problems of ever-larger sizes.
From now on, the reader should imagine that for each $n$, we have an instance of the house allocation problem
of size $n$; most quantities will accordingly have $n$ as a subscript.

In the upcoming sections, we shall need to consider the fortunes of agents as functions of their position
in the choosing order $\rho$. 

\begin{defn}
\label{defn:theta}
Define the {\em relative position} of an agent $a$ in the order $\rho$ to be
the fraction of all the agents whose position in $\rho$ is no worse than that of $a$. Thus, the
first agent in $\rho$ has relative position $1/n$ and the last has relative position $1$.
For $0\leq\theta\leq1$, let $A_n(\theta)$ denote the set of agents whose relative position is at most
$\theta$, and let $a_n(\theta)$ be the last agent in $A_n(\theta)$.

\end{defn}
\begin{remark}

For completeness, when $\theta < 1/n$ we let $a_n(\theta)$ be the first agent in $\rho$.
This exceptional definition will cause no trouble, as for $\theta>0$ it applies to only finitely many
$n$ and so does not affect asymptotic results, while for $\theta=0$ it allows us to say something
about the first agent in $\rho$.
\end{remark}

\section{Naive Boston}
\label{s:naive_asymptotics}

We now consider the Naive Boston algorithm. We begin with results about initial segments of the queue of agents.

\subsection{Groups of agents}
\label{ss:naive_all}

It will be useful to define the following sequence.

\begin{defn}
The sequence $(\omega_r)_{r=1}^{\infty}$ is defined by the initial condition  $\omega_1 = 1$ and 
recursion $\omega_{r+1}=\omega_r e^{-\omega_r}$ for $r\geq 1$.
\end{defn}
Thus, for example, $\omega_1 = 1, \omega_2 = e^{-1}, \omega_3 = e^{-1} e^{-e^{-1}}$.
The value of $\omega_r$ approximates $r^{-1}$, a relationship made more precise in the following result.

\begin{lemma}
\label{lem:wr_asymptotics}
For all $r\geq3$,
$$ \frac{1}{r+\log r} \;<\; \omega_r \;<\; \frac{1}{r}.
  $$
\end{lemma}

\begin{proof}{Proof.}
For $3\leq r \leq 8$ the inequalities can be verified by direct calculation. Beyond this, we rely on induction:
assume the result for a given $r\geq8$ and consider $\omega_{r+1}$. Observe that the function $x\mapsto xe^{-x}$
is monotone increasing on $[0,1]$: this gives us
$$ \omega_{r+1} \;=\; \omega_r e^{-\omega_r} \;<\; \frac{e^{-1/r}}{r}
 \;=\; \frac{1}{r+1} \exp\left(\log\left(1+\frac1r\right) - \frac1r\right)
 \;\leq\; \frac{1}{r+1},
  $$
via the well-known inequality $\log(1+x)\leq x$.
Also,
\begin{eqnarray*}
\omega_{r+1} \;=\; \omega_r e^{-\omega_r}
&>& \frac{e^{-1/(r+\log r)}}{r+\log r} \\
&\geq& \frac{1}{r+1+\log(r+1)}\left(1 + \frac{1+\log(r+1)-\log r}{r+\log r}\right)\left(1 - \frac{1}{r+\log r}\right),
\end{eqnarray*}
via the well-known inequality $e^{-x}\geq1-x$. Thus
\begin{eqnarray*}
\omega_{r+1}
&>& \frac{1}{r+1+\log(r+1)}\left(1 + \frac{(r-1+\log r)(\log(r+1)-\log r)-1}{(r+\log r)^2}\right) \\
&>& \frac{1}{r+1+\log(r+1)}\left(1 + \frac{-2+\log r}{(r+1)(r+\log r)^2}\right),
\end{eqnarray*}
since 
$$\log(r+1)-\log r=\int_r^{r+1}t^{-1}\, dt>\frac{1}{r+1}.$$
For $r\geq8$ we have $\log r > 2$ and so the result follows.
\hfill \qedsymbol
\end{proof}

We can now state our main result on the asymptotics of naive Boston.

\begin{theorem}[Number of agents remaining]
\label{thm:naive_asym_agent_numbers}
Consider the naive Boston algorithm. Fix $r\geq1$ and a relative position $\theta\in[0,1]$.
Then the number $N_n(r,\theta)$ of members of $A_n(\theta)$ present at round $r$ satisfies
$$ \frac1n N_n(r,\theta) \cvginprob z_r(\theta).
  $$
where $z_1(\theta)=\theta$ and
\begin{equation}
\label{eq:zr_recursion}
 z_{r+1}(\theta) = z_r(\theta)  - \left(1 - e^{-z_r(\theta)}\right)\omega_r \qquad\text{ for $r\geq1$.}
\end{equation}
In particular, the total number $N_n(r)$ of agents (and of items) present at round $r$ satisfies
$$ \frac1n N_n(r) \cvginprob z_r(1) = \omega_r .
  $$
\end{theorem}

Some of the functions $z_r(\theta)$ are illustrated in Figure \ref{fig:Boston_agent_numbers}.
Note that agents with an earlier position in $\rho$ are more likely to exit in the early rounds.
A consequence is that the position of an unsuccessful agent relative to other unsuccessful agents
tends to improve each time he fails to claim an item.

\begin{figure}[hbtp]
\centering
\includegraphics[width=0.8\textwidth]{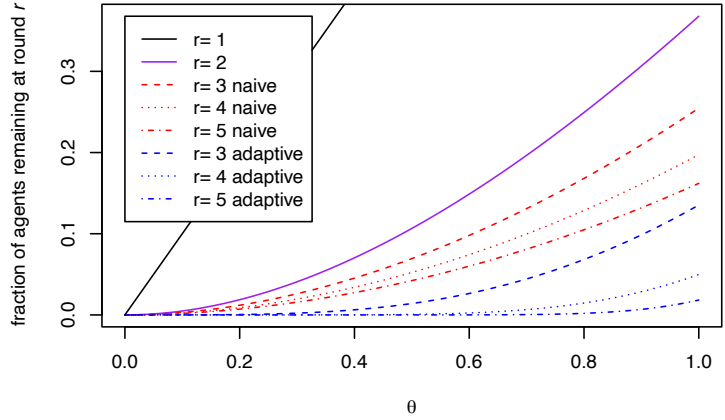}
\caption{The limiting fraction of the agents who have relative position $\theta$ or better and survive to participate in the $r$th round.}
\label{fig:Boston_agent_numbers}
\end{figure}

A better understanding of the functions $z_r(\theta)$ is given by the following result.

\begin{theorem}
\label{thm:naive_zr}
The functions $z_r(\theta)$ satisfy $z_r(\theta)=\int_0^\theta z'_r(\phi)\;d\phi$, where
\begin{eqnarray}
\label{eq:zr_prime_recursion}
z'_r(\theta) &=& \prod_{k=1}^{r-1} f_k(\theta) \qquad\qquad\hbox{ for $r\geq2$, and }\qquad z'_1(\theta)=1\\
f_r(\theta) &=& 1 - \omega_r \exp\left(-z_r(\theta)\right) .
\end{eqnarray}
\end{theorem}

\begin{proof}{Proof.}
Differentiate (\ref{eq:zr_recursion}) with respect to $\theta$.
Alternatively, integrate (\ref{eq:zr_prime_recursion}) by parts.
\qedsymbol
\end{proof}

\smallbreak

The quantity $f_r(\theta)$ can be interpreted (in a sense to be made precise later) as the conditional probability
that an agent with relative position $\theta$, if present at round $r$, is unmatched at that round.
The quantity $z'_r(\theta)$ can then be interpreted as the probability that an agent with relative position $\theta$
is still unmatched at the beginning of round $r$.
For the particular case of the last agent, we may note that $f_r(1) = 1-\omega_{r+1}$.
Other quantities for the first few rounds are shown in Table \ref{t:naive_limits}.


\begin{table}[hbtp]
\centering
\begin{tabular}{|l|c|c|c|c|}
\hline
meaning at round $r$ & quantity & $r=1$ & $r=2$ & $r=3$\\
\hline
Fraction of all agents: &&&&\\
$\bullet$ present & $\omega_r$ & $1$ & $e^{-1}\approx0.3679$ & $\exp(-1-e^{-1})\approx0.2546$\\
$\bullet$ in $A_n(\theta)$ and present & $z_r(\theta)$ & $\theta$ & $\theta+e^{-\theta}-1$
    & $\theta+e^{-\theta}-1-e^{-1}+\exp(-\theta-e^{-\theta})$\\
For an agent with relative position $\theta$: &&&&\\
$\bullet$ P(present) & $z'_r(\theta)$ & $1$ & $1-e^{-\theta}$
    & $(1-e^{-\theta})(1-\exp(-\theta-e^{-\theta}))$\\
$\bullet$ P(unmatched|present) & $f_r(\theta)$ & $1-e^{-\theta}$ & $1-\exp(-\theta-e^{-\theta})$
    & $1-\exp(-\theta-e^{-\theta}-\exp(-\theta-e^{-\theta}))$\\
\hline
\end{tabular}
\vspace{10pt}
\caption{Limiting quantities as $n\to\infty$ for the early rounds of the Naive Boston algorithm.}
\label{t:naive_limits}
\end{table}

\begin{theorem}
\label{thm:zbounds}
For $r\geq2$ and $0\leq\theta\leq1$,
\begin{equation}
\label{eq:zbound1}
 c_1 \omega_r (1-e^{-\theta}) \;\leq\; z'_r(\theta) \;\leq\; c_2 \omega_r (1-e^{-\theta})
\end{equation} 
and
\begin{equation}
\label{eq:zbound2}
 c_1 \omega_r (\theta + e^{-\theta} - 1) \;\leq\; z_r(\theta) \;\leq\; c_2 \omega_r (\theta + e^{-\theta} - 1)
\end{equation}
where the constants $c_1=e-1\approx1.718$ and $c_2=\exp(1+e^{-1})\approx2.927$.
\end{theorem}

\begin{proof}{Proof.}
It is enough to show (\ref{eq:zbound1}); (\ref{eq:zbound2}) then follows by integration.
From (\ref{eq:zr_prime_recursion}) we have
$$ f_1(\theta) \prod_{k=2}^{r-1} \left(1-\omega_r e^{-z_r(0)}\right)
\;\leq\; z'_r(\theta)
\;\leq\; f_1(\theta) \prod_{k=2}^{r-1} \left(1-\omega_r e^{-z_r(1)}\right)
  $$
since $z_r(\theta)$ is increasing in $\theta$. We have $f_1(\theta)=1-e^{-\theta}$, $z_r(0)=0$, and $z_r(1)=\omega_r$, so
\begin{equation}
\label{eq:zbound3}
(1-e^{-\theta}) \prod_{k=2}^{r-1} \left(1-\omega_r\right)
\;\leq\; z'_r(\theta)
\;\leq\; (1-e^{-\theta}) \prod_{k=3}^{r} \left(1-\omega_r\right) .
\end{equation}
Let $L_r=\omega_r^{-1}\prod_{k=2}^{r} \left(1-\omega_r\right)$ for all $r\geq1$.
This is an increasing sequence,
since $L_{r+1}/L_r = e^{\omega_r}(1-\omega_{r+1}) = e^{\omega_r} - \omega_r > 1$.
Hence, $L_r\geq L_2=\omega_2^{-1}(1-\omega_2)=e-1$ for all $r\geq2$;
that is, $\prod_{k=2}^{r} \left(1-\omega_r\right)\geq(e-1)\omega_r$ for $r\geq2$.
The lower bound in (\ref{eq:zbound3}) can thus be replaced by
$(e-1) \omega_{r-1} (1-e^{-\theta}) \;\leq\; z'_r(\theta)$ when $r\geq3$.
Since $\omega_{r-1}>\omega_r$, we obtain the lower bound in (\ref{eq:zbound1}) for $r\geq3$,
and we may verify directly that $z'_2(\theta)=1-e^{-\theta}$ satisfies this bound also.

A similar argument suffices for the upper bounds.
Let $U_r=\omega_{r+1}^{-1}\prod_{k=3}^r \left(1-\omega_r\right)$ for $r\geq2$.
This is a decreasing sequence,
since $U_r/U_{r-1}=e^{\omega_r}(1-\omega_r)<e^{\omega_r}e^{-\omega_r}=1$.
Hence, $U_r\leq U_2=\omega_3^{-1}$ for all $r\geq2$;
that is, $\prod_{k=3}^r \left(1-\omega_r\right) \leq \omega_3^{-1}\omega_{r+1}$ for $r\geq2$.
The upper bound in (\ref{eq:zbound3}) can thus be replaced by
$z'_r(\theta) \;\leq\; \omega_3^{-1}\omega_{r+1} (1-e^{-\theta})$ for $r\geq2$.
The constant $\omega_3^{-1}=\omega_2^{-1}e^{\omega_2}=\exp(1+e^{-1})$.
Since $\omega_{r+1}<\omega_r$, we obtain the upper bound in (\ref{eq:zbound1}).
\qedsymbol
\end{proof}

\begin{proof}{Proof of Theorem \ref{thm:naive_asym_agent_numbers}.}
Induct on $r$. For $r=1$, the result is immediate because $N_n(1,\theta)=\lfloor n\theta\rfloor$.
Now fix $r\geq1$ and assume the result for round $r$.
Let $\F_r$ be the $\sigma$-field generated by events prior to round $r$. Conditional on $\F_r$,
we have the situation of Lemma \ref{lem:Aprime}: there are $N_n(r)$ available items and
$N_n(r,\theta)$ agents of $A_n(\theta)$ who will be the first to attempt to claim them, with the
agents' bids chosen iid uniform from a larger pool of $n-r+1$ items.
Letting $S_n$ denote the number of these agents whose bids are successful, Lemma \ref{lem:Aprime} gives
$$ \hbox{Var}(S_n|\F_r) \;\leq\; E[S_n|\F_r]
\;=\; \frac{N_n(r)}{n-r+1} \sum_{a=1}^{N_n(r,\theta)} \left(1 - \frac{1}{n-r+1}\right)^{a-1} .
  $$
Summing the geometric series,
$$ E[S_n|\F_r] \;=\; 
  N_n(r) \left(1-\left(1-\frac{1}{n-r+1}\right)^{N_n(r,\theta)}\right) .
  $$
It then follows by the inductive hypothesis that
$$ \frac1n E[S_n|\F_r] \;\cvginprob\; \omega_r \left(1 - e^{-z_r(\theta)}\right) \qquad\hbox{ as $n\to\infty$.}
  $$
By Lemma \ref{lem:CC},
$$ \frac1n S_n \;\cvginprob\; \omega_r \left(1 - e^{-z_r(\theta)}\right) .
  $$
We have $N_n(r+1,\theta) = N_n(r,\theta) - S_n$, and so obtain
$$ \frac1n N_n(r+1,\theta) \;\cvginprob\; z_r(\theta) - \omega_r \left(1 - e^{-z_r(\theta)}\right)
= z_{r+1}(\theta).
  $$
The result follows.
\qedsymbol
\end{proof}

\begin{corollary}[limiting distribution of preference rank obtained]
\label{cor:naive_obtained}
The number $S_n(s,\theta)$ of members of $A_n(\theta)$ matched to their $s$th preference satisfies
$$ \frac1n S_n(s,\theta) \cvginprob \int_0^\theta q_s(\phi)\;d\phi .
  $$
where 
$$ q_s(\theta) = z'_s(\theta) - z'_{s+1}(\theta)  = z'_s(\theta) \omega_s e^{-z_s(\theta)}.
  $$
\end{corollary}

\begin{proof}{Proof.}
An agent is matched to his $s$th preference if, and only if, he is present at round $s$ but not at round $s+1$.
The result follows by Theorem \ref{thm:naive_asym_agent_numbers}.
\qedsymbol
\end{proof}

The limiting functions $q_s(\theta)$ are illustrated in Figure \ref{fig:naive_exit_round}. For example, 
an agent at relative position $1/2$ has 
probability over 78\% of exiting at the first round while the last agent has corresponding probability just under $37\%$.

\begin{figure}[hbtp]
\centering
\includegraphics[width=0.8\textwidth]{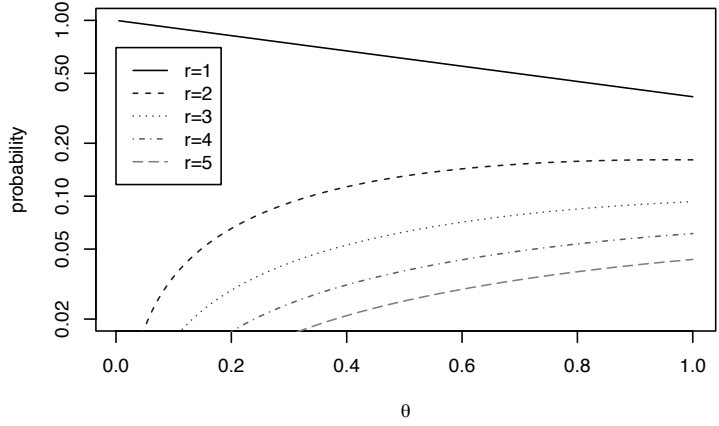}
\caption{The limiting probability that an agent exits the naive Boston mechanism at the $r$th round
(and so obtains his $r$th preference), as a function of the agent's initial relative position $\theta$.
Logarithmic scale on vertical axis.}
\label{fig:naive_exit_round}
\end{figure}

\subsection{Individual agents}
\label{ss:naive_individual}

Theorem \ref{thm:naive_asym_agent_numbers} and Corollary \ref{cor:naive_obtained} are concerned with
the outcomes achieved by the agent population collectively, and will be used in  Section~\ref{s:welfare} to say
something about utilitarian welfare.

Suppose, though, that our interest lies with individual agents.
It is tempting to informally ``differentiate" the result of
Theorem \ref{thm:naive_asym_agent_numbers} with respect to $\theta$, and thereby draw conclusions
about the fate of a single agent. The following result puts those conclusions on a sound footing.

\begin{theorem}[exit time of individual agent]
\label{thm:naive_individual_asymptotics}
Consider the naive Boston algorithm. Fix $r\geq1$ and a relative position $\theta\in[0,1]$.
Let $R_n(\theta)$ denote the round number at which the agent $a_n(\theta)$
(the last agent with relative position at most $\theta$) is matched.
Equivalently, $R_n(\theta)$ is the preference rank of the item obtained by this agent.
Then 
$$P(R_n(\theta)\geq r)\to z'_r(\theta) \quad \text{as $n\to\infty$}.
$$
\end{theorem}

\begin{remark}
\label{rem:naive_qs}
The result of Theorem \ref{thm:naive_individual_asymptotics} could equivalently be stated as 
$$ P(R_n(\theta)=r)\to q_r(\theta)
  $$
where $q_r(\theta)$ is as in Corollary \ref{cor:naive_obtained}.
Note that $\sum_{r=1}^{\infty} q_r(\theta)=1$, consistent with the role of
$q_r(\theta)$ as an asymptotic probability.
\end{remark}

\begin{remark}
Theorem \ref{thm:naive_individual_asymptotics} tells us that an agent with fixed relative
position $\theta$ has a good chance of obtaining one of his first few preferences,
even if $n$ is large. This is even true of the very last agent ($\theta=1$). 
Figure~\ref{fig:naive_exit_round} displays the limiting values. 
\end{remark}

\begin{proof}{Proof of Theorem \ref{thm:naive_individual_asymptotics}.}
The result is trivial for $r=1$. Assume the result for a given value of $r$, and let
$\F_r$ be the $\sigma$-field generated by events prior to round $r$. Conditional on $\F_r$,
we can apply Lemma \ref{lem:Aprime} to the single agent $a_n(\theta)$ to obtain
\begin{equation}
\label{eq:naive_individual_induction}
P(R_n(\theta)\geq r+1)
 \;=\; E\left[P(R_n(\theta)\geq r+1|\F_r)\right]
 \;=\; E\left[1_{R_n(\theta)\geq r} Y_n\right] ,
\end{equation} 
where
$$ Y_n \;=\; 1 \;-\; \left(1 - \frac{1}{n-r+1}\right)^{N_n(r,\theta)-1} \left(\frac{N_n(r)}{n-r+1}\right) .
  $$
Observe that $Y_n\cvginprob 1 - \omega_r e^{-z_r(\theta)} = f_r(\theta)$ from Theorem \ref{thm:naive_asym_agent_numbers}.
Equation (\ref{eq:naive_individual_induction}) gives
$$ P(R_n(\theta)\geq r+1) \;-\; z'_{r+1}(\theta)
 \;=\; E\left[1_{R_n(\theta)\geq r}(Y_n - f_r(\theta))\right]
   \;+\; E\left[1_{R_n(\theta)\geq r} - z'_r(\theta)\right] f_r(\theta) .
  $$
The second term converges to 0 as $n\to\infty$ by the inductive hypothesis.
For the first term, note that the convergence $Y_n - f_r(\theta) \cvginprob 0$ is also
convergence in ${\mathcal L}_1$ by Theorem 4.6.3 in \cite{Durrett}, and so
$1_{R_n(\theta)\geq r}(Y_n - f_r(\theta))\to 0$ in ${\mathcal L}_1$ also.
\hfill \qedsymbol
\end{proof}

\section{Adaptive Boston}
\label{s:adaptive_asymptotics}

We again begin with results about initial segments of the queue of agents, and follow up with results about individual agents.

\subsection{Groups of agents}
\label{ss:adaptive_all}

A simple stochastic model of IC bidding for the adaptive Boston mechanism can be similar to the naive case.
At the beginning of the $r$th round, each remaining agent randomly chooses an item as his next preference for
which to bid; the bid is successful, and the agent matched to that item, if no other agent with an earlier
position in the order $\rho$ bids for the same item. But, whereas a naive-Boston participant chooses
from the set of $n-r+1$ items for which he has not already bid, the adaptive-Boston participant chooses from
a smaller set: the $N_n(r)$ items actually still available at the beginning of the round.
This model allows a result analogous to Theorem \ref{thm:naive_asym_agent_numbers}.

\begin{theorem}[Number of agents remaining]
\label{thm:adaptive_asym_agent_numbers}
Consider the adaptive Boston algorithm. Fix $r\geq1$ and a relative position $\theta\in[0,1]$.
Then the number $N_n(r,\theta)$ of members of $A_n(\theta)$ present at round $r$ satisfies
$$ \frac1n N_n(r,\theta) \cvginprob y_r(\theta).
  $$
where $y_1(\theta)=\theta$ and
\begin{equation}
\label{eq:yr_recursion}
y_{r+1}(\theta) = y_r(\theta) - e^{1-r}\left(1 - \exp\left(-e^{r-1}y_r(\theta)\right)\right) \qquad\text{ for $r\geq1$.}
\end{equation}
In particular, the total number $N_n(r)$ of agents (and of items) present at round $r$ satisfies
$$ \frac1n N_n(r) \cvginprob y_r(1) = e^{1-r} .
  $$
\end{theorem}

Some of the functions $y_r(\theta)$ are illustrated in Figure \ref{fig:Boston_agent_numbers}.
It is apparent that the adaptive Boston mechanism proceeds more quickly than naive Boston:
$e^{1-r}$ decays much more quickly than $\omega_r$ as $r\to\infty$.
Also, the tendency of advantageously-ranked agents to be matched in relatively early rounds is even
greater for the adaptive version of the algorithm.
In an adaptive-Boston assignment of a large number of items to agents with IC preferences,
under 2\% of the agents will be unmatched after four rounds (vs. 16\% for naive Boston),
and most of these (about 2/3) will be among the last 10\% of agents in the original
agent order.

A better understanding of the functions $y_r(\theta)$ is given by the following result, which is analogous to Theorem~\ref{thm:naive_zr}.

\begin{theorem}
\label{thm:adaptive_yr}
The functions $y_r(\theta)$ satisfy $y_r(\theta)=\int_0^\theta y'_r(\phi)\;d\phi$, where
\begin{eqnarray}
\label{eq:yr_prime_recursion}
y'_r(\theta) &=& \prod_{k=1}^{r-1} g_k(\theta) \qquad\hbox{ for $r\geq2$}, \quad\hbox{ and $y'_1(\theta)=1$}\\
g_r(\theta) &=& 1 - \exp\left(-e^{r-1} y_r(\theta)\right)
\end{eqnarray}
\end{theorem}

\begin{proof}{Proof.}
Differentiate (\ref{eq:yr_recursion}) with respect to $\theta$.
Alternatively, integrate (\ref{eq:yr_prime_recursion}) by parts.
\qedsymbol
\end{proof}

\smallbreak

\begin{remark}
\label{rem:adaptive_yr}
The quantity $g_r(\theta)$ is analogous to $f_r(\theta)$ in the naive case, and can be interpreted
(in a sense to be made precise later) as the conditional probability that an agent with relative position
$\theta$, if present at round $r$, is unmatched at that round.
The quantity $y'_r(\theta)$, analogous to $z'_r(\theta)$ in the naive case, can then be interpreted
as the probability that an agent with relative position $\theta$ is still unmatched at the beginning of round $r$.
For the particular case of the last agent, we may note that $g_r(1) = 1 - e^{-1}$ and $y'_r(1) = (1-e^{-1})^{r-1}$.
Other quantities for the first two rounds are shown in Table \ref{t:adaptive_limits}.
\end{remark}


\begin{table}[hbtp]
\centering
\begin{tabular}{|l|c|c|c|c|}
\hline
meaning at round $r$ & quantity & $r=1$ & $r=2$\\
\hline
Fraction of all agents: &&&\\
$\bullet$ present & $e^{1-r}$ & $1$ & $e^{-1}\approx 0.3679$ \\
$\bullet$ in $A_n(\theta)$ and present & $y_r(\theta)$ &$\theta$  &$\theta + e^{-\theta} - 1$\\
For an agent with relative position $\theta$: &&&\\
$\bullet$ P(present) & $y'_r(\theta)$ & $1$ & $1 - e^{-\theta}$\\
$\bullet$ P(unmatched|present) & $g_r(\theta)$ & $1 - e^{-\theta}$
     & $1-\exp(-e(\theta+e^{-\theta}-1))$\\
$\bullet$ P(bids for $s$th preference|present) & $u_{rs}$ & $1_{s=1}$
     & $e^{-1}(1-e^{-1})^{s-2}1_{s\geq2}$ \\
\hline
\end{tabular}
\vspace{10pt}
\caption{Limiting quantities for the early rounds of the adaptive Boston algorithm.}
\label{t:adaptive_limits}
\end{table}

\begin{proof}{Proof of Theorem \ref{thm:adaptive_asym_agent_numbers}.}
Induct on $r$.
For $r=1$ we have $N_n(1,\theta)=\lfloor n\theta\rfloor$; the result follows immediately.
Now suppose the result for a given value of $r$, and consider $r+1$.
Let $T_n$ be the number of agents of $A_n(\theta)$ matched at round $r$.
Conditioning on the $\sigma$-field $\F_r$ generated by events prior to round $r$, we have the situation
of Lemma \ref{lem:A}: there are $N_n(r)$ available items and $N_n(r,\theta)$ agents of $A_n(\theta)$
who will be the first to attempt to claim them, with each such agent bidding for one of the available items,
chosen uniformly at random independently of other agents.
Lemma \ref{lem:A} gives us $\hbox{Var}(T_n|\F_r)\leq E[T_n|\F_r]$ and
$$ E[T_n|\F_r]
\;=\; \sum_{a=1}^{N_n(r,\theta)} \left(1-\frac{1}{N_n(r)}\right)^{a-1}
\;=\; N_n(r) \left(1 - \left(1-\frac{1}{N_n(r)}\right)^{N_n(r,\theta)} \right) .
  $$
By the inductive hypothesis,
$$ \frac{N_n(r)}{n} \;\cvginprob\; e^{1-r}
\qquad\hbox{ and }\qquad
\left(1-\frac{1}{N_n(r)}\right)^{N_n(r,\theta)} \;\cvginprob\; \exp\left(-e^{r-1}y_r(\theta)\right) .
  $$
This gives us
$$ \frac1n E[T_n|\F_r]
\;\cvginprob\; e^{1-r}\left(1-\exp\left(-e^{r-1}y_r(\theta)\right)\right)
\;=\; y_r(\theta) - y_{r+1}(\theta) .
  $$
By Lemma \ref{lem:CC}, then,
$$ \frac1n T_n 
\;\cvginprob\; y_r(\theta) - y_{r+1}(\theta) .
  $$
Since $T_n=N_n(r,\theta) - N_n(r+1,\theta)$, it follows that
$\frac1n N_n(r+1,\theta)\cvginprob y_{r+1}(\theta)$.
Hence the result.  
\qedsymbol
\end{proof}

\subsubsection*{The rank of the item received}

Theorem \ref{thm:adaptive_asym_agent_numbers} is less satisfying than Theorem \ref{thm:naive_asym_agent_numbers}.
The naive Boston mechanism has a key simplifying feature: the rank of an item within its assigned
agent's preference order is equal to the round number in which it was matched. This means that Theorem
\ref{thm:naive_asym_agent_numbers} already enables some conclusions about agents' satisfaction with the outcome
of the process (see Corollary \ref{cor:naive_obtained}). But, in the adaptive case, we know only that an item matched at round $r>1$ will be
no better (and could be worse) than its assigned agent's $r$th preference.

To do better, we need a more detailed stochastic bidding model.
An agent $a$ still present at the beginning of the $r$th round will have thus far determined an initial
sub-sequence of his preference order comprising some number $F_{a,r-1}$ of most-preferred items, and failed
to obtain any of them. He thus has a pool of $n-F_{a,r-1}$ previously-unconsidered items from which to choose,
of which the $N_n(r)$ items actually still available are a subset. In accordance with the IC model, let us
imagine that he now generates further preferences by repeated random sampling without replacement from the
previously-unconsidered items, until one of the available items is sampled; this item becomes his bid in
the current round. Denote by $G_{ar}$ the number of items sampled to construct this bid;
thus $F_{ar}=\sum_{j=1}^{r} G_{aj}$ and $G_{a,1}=1$.
If the bid is successful, the agent will be matched to his $F_{ar}$th preference.

Note that while the simple bidding model used in Theorem \ref{thm:adaptive_asym_agent_numbers}
provides enough information to determine the matching of items to agents (along with the round numbers
at which the items are matched), it does not completely determine the agents' preference orders.
In particular, it does not determine the agents' preference ranks for the items they are assigned.
The random variables $G_{ar}$ provide additional information sufficient to determine this interesting
feature of the outcome.

It is convenient to think of the $G_{ar}$ and $F_{ar}$ as being determined by an auxiliary process
that runs after the simple bidding model has been run and the matching of agents to items determined.
This auxiliary process can be described in the following way.
Fix integers $n_1>n_2>\cdots>n_r>0$.
\begin{itemize}
\item Place $n_1$ balls, numbered from 1 to $n_1$, in an urn.
\item For $i=1,\ldots,r$
\begin{itemize}
\item Deem the $n_i$ lowest-numbered balls remaining in the urn ``good".
\item Draw balls at random from the urn, without replacement, until a good ball is drawn.
\end{itemize}
\end{itemize}
Let $H(n_1,\ldots,n_r)$ be the probability distribution of the total number of balls drawn, and
$q(s;n_1,\ldots,n_r) = P(X=s)$ where $X\sim H(n_1,\ldots,n_r)$.

Denote by $\M$ the $\sigma$-field generated by the simple bidding model, including the items on which each
agent bids and the resulting matching. Conditional on $\M$, the random variable $F_{ar}$ for an agent
$a$ still present at round $r$ has the $H(n,N_n(2),\ldots,N_n(r))$ distribution.
That is,
\begin{equation}
\label{eq:adaptive_bidding_M}
P(F_{ar}=s|\M) \;=\; q(s;n,N_n(2),\ldots,N_n(r)) .
\end{equation}
Also, the $\{F_{ar}:a\hbox{ present at round }r\}$ are conditionally independent given $\M$.

\begin{lemma}
\label{lem:H_distn}
$q(1;n_1)=1$; $q(s;n_1)=0$ for $s>1$; and $q(s;n_1,\ldots,n_r)=0$ for $s<r$ or $s>n_1-n_r+1$.
The $H(n_1,\ldots,n_r)$ distribution's other probabilities are given by the recurrence
$$ q(s;n_1,\ldots,n_r)
 \;=\; \sum_{t=r-1}^{s-1} q(t;n_1,\ldots,n_{r-1}) \left(\frac{n_r}{n_1-s+1}\right)
       \prod_{0\leq i<s-t-1} \left(1-\frac{n_r}{n_1-t-i}\right) .
  $$
\end{lemma}

\begin{proof}{Proof.}
Let $N$ be the number of balls drawn in the first $r-1$ iterations of the process, and $M$ the number drawn
in the final iteration. Then $P(N+M=s) \;=\; \sum_{t=r-1}^{s-1} P(N=t) P(M=s-t|N=t)$, and we have
$$ P(M=s-t|N=t) \;=\; \left(\frac{n_r}{n_1-s+1}\right) \prod_{0\leq i<s-t-1} \left(1-\frac{n_r}{n_1-t-i}\right) .
  $$
(The final iteration must first sample $s-t-1$ consecutive non-good balls: the probabilities of achieving
this are $1-\frac{n_r}{n_1-t}$ for the first, $1-\frac{n_r}{n_1-t-1}$ for the second,
$\ldots1-\frac{n_r}{n_1-s+2}$ for the last. At last, a good ball must be drawn:
the probability of this is $\frac{n_r}{n_1-s+1}$.)
The result follows.
\qedsymbol
\end{proof}

\smallbreak

Our interest in the $H(n_1,\ldots,n_r)$ distribution mostly concerns its asymptotic limits as the numbers of balls
become large, and the ``without replacement'' stipulation becomes unimportant.
To this end, fix $p_1,\ldots,p_r\in(0,1]$ and let $u(s;p_1,\ldots,p_r)=P\left(r+\sum_{i=1}^r G_i = s\right)$,
where $G_1,\ldots,G_r$ are independent random variables
with geometric distributions: $P(G_i=x)=p_i(1-p_i)^x$ for $x=0,1,\ldots$.

\begin{lemma}
\label{lem:G_distn}
$u(s;p)=p(1-p)^{s-1}$; $u(s;p_1,\ldots,p_r)=0$ for $s<r$; and
$$ u(s;p_1,\ldots,p_r) \;=\; \sum_{t=r-1}^{s-1} u(t;p_1,\ldots,p_r) p_r (1-p_r)^{s-t-1} . 
  $$
\end{lemma}

\begin{proof}{Proof.}
$$ P\left(r+\sum_{i=1}^r G_i = s\right) \;=\; \sum_{t=r-1}^{s-1} P\left(r-1 +\sum_{i=1}^{r-1} G_i = t\right) P(1+G_r = s-t) . \qedsymbol
  $$
\end{proof}

\begin{lemma}
$$ q(s;n_1,\ldots,n_r) \;\to\; u(s;p_1,\ldots,p_r)
\qquad\text{as $n_1,\ldots,n_r\to\infty$ with $\frac{n_i}{n_1}\to p_i$.}
  $$
\end{lemma}

\begin{proof}{Proof.}
Take limits in Lemma \ref{lem:H_distn}; compare Lemma \ref{lem:G_distn}.
\qedsymbol
\end{proof}

\begin{corollary}
\label{cor:F_cvginprob}
Consider the adaptive Boston mechanism, and fix $s$. We have
$$ q(s;n,N_n(2),\ldots,N_n(r)) \;\cvginprob\; u(s;1,e^{-1},\ldots,e^{1-r})
\qquad\text{as $n\to\infty$.}
  $$
\end{corollary}

\begin{proof}{Proof.}
Use the convergence of $\frac1n N_n(i)$ given by Theorem \ref{thm:adaptive_asym_agent_numbers}.
\qedsymbol
\end{proof}

Corollary \ref{cor:F_cvginprob} and (\ref{eq:adaptive_bidding_M}) give us an asymptotic limit
for the distribution, conditional on $\M$, of $F_{ar}$, the preference rank of the bid made at
round $r$ by an agent still present at that round.
To condense notation, we will denote the limit $u(s;1,e^{-1},\ldots,e^{1-r})$ by $u_{rs}$.
That is,
$$ P(F_{ar}=s|\M) \;\cvginprob\; u_{rs} .
  $$
Note that the limit $u_{rs}$ does not depend on the position of the agent $a$ in the choosing order.
It is fairly clear why this should be so: all remaining agents must enter their bids at the
beginning of the round, before any other agent has bid, and so the bidding process, at least,
treats them symmetrically. The advantage arising from a favourable position lies in a higher
probability of obtaining the item bid for, not in constructing the bid itself.

We make use of the following simplified recurrence.

\begin{lemma}
\label{lem:three_term_recurrence}
$u(s;p)=p(1-p)^{s-1}$ and $u(s;p_1,\ldots,p_r)=0$ for $s<r$; other values are given by the recurrence
$$ u(s;p_1,\ldots,p_r) \;=\; p_r u(s-1;p_1,\ldots,p_{r-1}) \;+\; (1-p_r) u(s-1;p_1,\ldots,p_{r}) .
  $$
In particular: $u_{11}=1$, $u_{1,s}=0$ for $s>1$, $u_{rs}=0$ for $s<r$, and
\begin{equation}
\label{eq:three_term_recurrence}
 u_{rs} \;=\; e^{1-r} u_{r-1,s-1} \;+\; (1-e^{1-r}) u_{r,s-1} .
\end{equation}
\end{lemma}

\begin{proof}{Proof.}
The recurrence in (\ref{eq:three_term_recurrence}) has a unique solution; as does the one in Lemma \ref{lem:G_distn}.
It is easy to check that either solution also satisfies the other recurrence.
\qedsymbol
\end{proof}

\begin{remark}
\label{rem:sum_urs}
It follows directly from \eqref{eq:three_term_recurrence} that the bivariate generating function $F(x,y) = \sum_{r,s} u_{rs} x^r y^s$ satisfies the 
defining equation $F(x,y)(1-y) = xy+F(x/e,y)(x-e)$. It follows directly (from substituting $y=1$) 
that
$\sum_{s=r}^{\infty} u_{rs} = 1$,
consistent with its role as a probability distribution. We have not found a nice 
explicit formula for $u_{rs}$.
\end{remark}

We can now state a more detailed version of Theorem \ref{thm:adaptive_asym_agent_numbers}.

\begin{theorem}[the bidding process at a given round]
\label{thm:adaptive_asymptotics}
Consider the adaptive Boston algorithm. Fix $s\geq r\geq1$ and a relative position $\theta\in[0,1]$.
Let $y_r(\theta)$ be as in Theorem \ref{thm:adaptive_asym_agent_numbers},
and $u_{rs}$ be as in Lemma \ref{lem:three_term_recurrence}.
\begin{enumerate}[(i)]
\item
\label{adaptive_bidding}
The number $N_n(r,s,\theta)$ of members of $A_n(\theta)$ making a bid for their $s$th preference at round $r$ satisfies
$$ \frac1n N_n(r,s,\theta) \cvginprob u_{rs} y_r(\theta)
  $$
\item
\label{adaptive_bidding_and_unmatched}
The number $U_n(r,s,\theta)$ of members of $A_n(\theta)$ making an unsuccessful bid for their $s$th preference at round $r$ satisfies
$$ \frac1n U_n(r,s,\theta) \cvginprob u_{rs} y_{r+1}(\theta).
  $$
\item
\label{adaptive_bidding_and_matched}
The number $S_n(r,s,\theta)$ of members of $A_n(\theta)$ making a successful bid for their $s$th preference at round $r$ satisfies
$$ \frac1n S_n(r,s,\theta) \cvginprob  u_{rs} (y_r(\theta) - y_{r+1}(\theta)).
  $$
\end{enumerate}
\end{theorem}

\begin{proof}{Proof.}
Conditional on the $\sigma$-field $\M$, each agent $a$ participating in round $r$ enters a bid for his $F_{ar}$th preference;
the $F_{ar}$ for this group of agents are conditionally independent given $\M$.
Thus, the conditional distribution of $N(r,s,\theta)$ given $\M$ is the binomial distribution with $N_n(r,\theta)$
trials and success probability $P(F_{ar}=s|\M)$ given by (\ref{eq:adaptive_bidding_M}).
The variance of a binomial distribution never exceeds its mean (\cite{Feller}), so Lemma \ref{lem:CC} applies.
We will thus obtain Part (\ref{adaptive_bidding}) of the theorem if we can merely show that
$\frac1n E[N_n(r,s,\theta)|\M] \cvginprob u_{rs} y_r(\theta)$; that is
\begin{equation}
\frac1n N_n(r,\theta) q(s;n,N_n(2),\ldots,N_n(r)) \cvginprob u_{rs} y_r(\theta) .
\end{equation}
Theorem \ref{thm:adaptive_asym_agent_numbers} gives $\frac1n N_n(r,\theta)\cvginprob y_r(\theta)$,
and Corollary \ref{cor:F_cvginprob} gives $q(s;n,N_n(2),\ldots,N_n(r))\cvginprob u_{rs}$.
Part (\ref{adaptive_bidding}) follows.

The proof of Part (\ref{adaptive_bidding_and_unmatched}) is very similar:
the conditional distribution of $U(r,s,\theta)$ given $\M$ is the binomial distribution with $N_n(r+1,\theta)$
trials and success probability $P(F_{ar}=s|\M)$ given by (\ref{eq:adaptive_bidding_M}).
Part (\ref{adaptive_bidding_and_matched}) follows from Parts (\ref{adaptive_bidding}) and (\ref{adaptive_bidding_and_unmatched}).
\qedsymbol
\end{proof}

We now have the analog for Adaptive Boston of Corollary~\ref{cor:naive_obtained}.

\begin{corollary}[limiting distribution of preference rank obtained]
\label{cor:adaptive_obtained}
The number $S_n(s,\theta)$ of members of $A_n(\theta)$ matched to their $s$th preference satisfies
$$ \frac1n S_n(s,\theta) \cvginprob \int_0^\theta q_s(\phi)\;d\phi .
  $$
where 
$$ q_s(\theta) = \sum_{r=1}^s u_{rs} \left(y'_r(\theta) - y'_{r+1}(\theta)\right)
               = \sum_{r=1}^s u_{rs} y'_r(\theta) \exp\left(-e^{r-1}y_r(\theta)\right).
  $$
\end{corollary}

\begin{figure}[hbtp]
\centering
\includegraphics[width=0.8\textwidth]{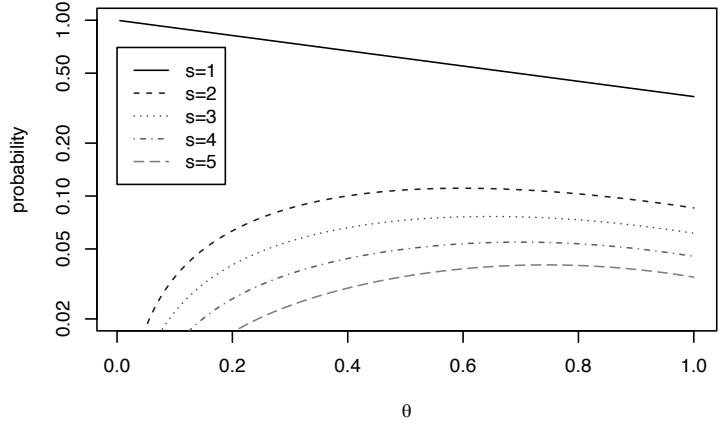}
\caption{The limiting probability $q_s(\theta)$ that an agent obtains his $s$th preference via the adaptive Boston
mechanism, as a function of the agent's initial relative position $\theta$. Logarithmic scale on vertical axis.}
\label{fig:adaptive_pref_obtained}
\end{figure}

The functions $q_s(\theta)$ are illustrated in Figure \ref{fig:adaptive_pref_obtained}.

Figure~\ref{fig:adaptive_round_2_last} shows for the last agent ($\theta = 1$) the distribution of the rank of the item bid for and the item obtained at the second round. 

\begin{figure}[hbtp]
\centering
\includegraphics[width=0.8\textwidth]{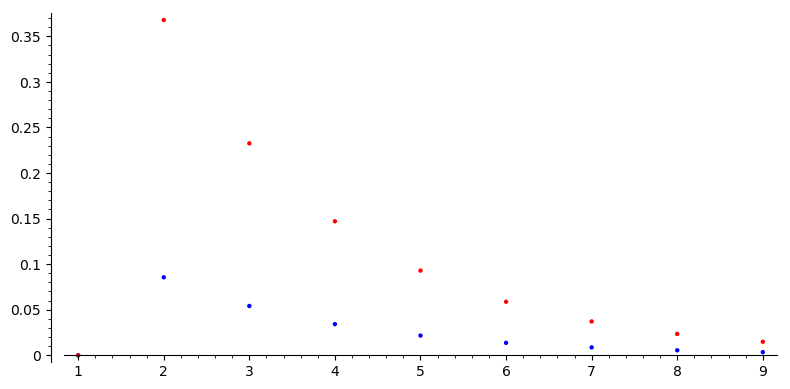}
\caption{Distribution of rank of item for which the last agent bids (upper) and successfully bids (lower) in round 2, Adaptive Boston}
\label{fig:adaptive_round_2_last}
\end{figure}


\begin{remark}
\label{rem:adaptive_sum_qs}
It is clear from the definition (and Remark \ref{rem:sum_urs}) that $\sum_{s=1}^{\infty} q_s(\theta)=1$.
This is consistent with the implied role of $q_s(\theta)$ as a probability distribution: the limiting
probability that an agent in position $\theta$ obtains his $s$th preference.
See also Theorem \ref{thm:adaptive_individual_asymptotics} Part \ref{adaptive_individual_obtained}.
\end{remark}


\begin{proof}{Proof of Corollary \ref{cor:adaptive_obtained}.}
This is an immediate consequence of Part (\ref{adaptive_bidding_and_matched}) of Theorem \ref{thm:adaptive_asymptotics},
with Theorem \ref{thm:adaptive_yr} providing the integral form of the limit.
\qedsymbol
\end{proof}

\subsection{Individual agents}
\label{ss:adaptive_individual}

If we wish to follow the fate of a single agent in the adaptive Boston mechanism, we need limits analogous
to that of Theorem \ref{thm:naive_individual_asymptotics}. These are provided by the following result.

\begin{theorem}[exit time and rank obtained for individual agent]
\label{thm:adaptive_individual_asymptotics}
Consider the adaptive Boston algorithm. Fix $s\geq r\geq1$ and a relative position $\theta\in[0,1]$.
Let $V_n(r,\theta)$ denote the preference rank of the item for which the agent $a_n(\theta)$
(the last agent with relative position at most $\theta$) bids at round $r$.
(For completeness, set $V_n(r,\theta)=0$ whenever $a_n(\theta)$ is not present at round $r$.)
Let $R_n(\theta)$ denote the round number at which $a_n(\theta)$ is matched.
Then
\begin{enumerate}
\item
\label{adaptive_individual_present}
(Agent present at round $r$.)
$$ P(R_n(\theta)\geq r)\to y'_r(\theta) \quad \text{as $n\to\infty$}.
  $$
\item
\label{adaptive_individual_bidding}
(Agent bids for $s$th preference at round $r$.)
$$ P(V_n(r,\theta)=s)\to y'_r(\theta) u_{rs} \quad \text{as $n\to\infty$}.
  $$
\item
\label{adaptive_individual_matched}
(Agent matched to $s$th preference at round $r$.)
$$ P(R_n(\theta)=r\hbox{ and }V_n(r,\theta)=s)\to y'_r(\theta) u_{rs} (1-g_r(\theta)) \quad \text{as $n\to\infty$}.
  $$
\item
\label{adaptive_individual_obtained}
(Agent matched to $s$th preference.)
$$ P(V_n(R_n(\theta),\theta)=s)\to q_s(\theta) \quad \text{as $n\to\infty$}.
  $$
\end{enumerate}
The limiting quantities $y'_r(\theta)$, $g_r(\theta)$, $u_{rs}$, and $q_s(\theta)$ are as defined in
Theorem \ref{thm:adaptive_yr}, Lemma \ref{lem:three_term_recurrence} and Corollary \ref{cor:adaptive_obtained}.
\end{theorem}

\begin{proof}{Proof.}
Part (\ref{adaptive_individual_present}) is proved in a similar way to Theorem \ref{thm:naive_individual_asymptotics}.
The result is trivial for $r=1$. Assume the result for a given value of $r$, and let $\F_r$ be the
$\sigma$-field generated by events prior to round $r$. Then
\begin{equation}
\label{eq:adaptive_individual_induction1}
P(R_n(\theta)\geq r+1)
 \;=\; E\left[P(R_n(\theta)\geq r+1|\F_r)\right]
 \;=\; E\left[1_{R_n(\theta)\geq r} Y_n\right] ,
\end{equation} 
where (by applying Lemma \ref{lem:A} to the single agent $a_n(\theta)$)
$$ Y_n \;=\; 1 \;-\; \left(1 - \frac{1}{N_n(r)}\right)^{N_n(r,\theta)-1} .
  $$
Observe that $Y_n\cvginprob 1 - \exp\left(-e^{r-1}y_r(\theta)\right) = g_r(\theta)$ by Theorem \ref{thm:adaptive_asym_agent_numbers}.
Equation (\ref{eq:adaptive_individual_induction1}) gives
$$ P(R_n(\theta)\geq r+1) \;-\; y'_{r+1}(\theta)
 \;=\; E\left[1_{R_n(\theta)\geq r}(Y_n - g_r(\theta))\right]
   \;+\; E\left[1_{R_n(\theta)\geq r} - y'_r(\theta)\right] g_r(\theta) .
  $$
The second term converges to 0 as $n\to\infty$ by the inductive hypothesis.
For the first term, note that the convergence $Y_n - g_r(\theta) \cvginprob 0$ is also
convergence in ${\mathcal L}_1$ by Theorem 4.6.3 in \cite{Durrett}, and so
$1_{R_n(\theta)\geq r}(Y_n - g_r(\theta))\to 0$ in ${\mathcal L}_1$ also.
Part (\ref{adaptive_individual_present}) follows.

For Part (\ref{adaptive_individual_bidding}), we have
\begin{equation*}
P(V_n(r,\theta)=s)
 \;=\; E\left[1_{R_n(\theta)\geq r} P(F_{a_n(\theta),r}=s|\M)\right]
 \;=\; E\left[1_{R_n(\theta)\geq r} q(s;n, N_n(2),\ldots,N_n(r))\right] .
\end{equation*}
Hence,
\begin{equation}
\label{eq:adaptive_individual_induction2}
P(V_n(r,\theta)=s) - u_{rs}y'_r(\theta)
 \;=\; E\left[1_{R_n(\theta)\geq r} (q(s;n, N_n(2),\ldots,N_n(r))-u_{rs})\right]
 \;+\; u_{rs}\left(P(R_n(\theta)\geq r) - y'_r(\theta)\right) .
\end{equation}
Both terms converge in probability to 0. For the second term, the convergence is given by
Part (\ref{adaptive_individual_present}). For the first term, it is a consequence of
Corollary \ref{cor:F_cvginprob}: $q(s;n,N_n(2),\ldots,N_n(r))\cvginprob u_{rs}$, which
is also convergence in ${\mathcal L}_1$ by Theorem 4.6.3 in \cite{Durrett}.
Part (\ref{adaptive_individual_bidding}) follows.

The proof of Part (\ref{adaptive_individual_matched}) is very similar to that of Part (\ref{adaptive_individual_bidding});
just replace $1_{R_n(\theta)\geq r}$ by $1_{R_n(\theta)=r}$
and $y'_r(\theta)$ by $y'_r(\theta) - y'_{r+1}(\theta)$.

Part (\ref{adaptive_individual_obtained}) is obtained from Part (\ref{adaptive_individual_matched})
by summation over $r$.
\qedsymbol
\end{proof}

\section{Serial Dictatorship}
\label{s:serial_dictatorship}

Unlike the Boston algorithms, SD is strategyproof, but it is known to behave worse in welfare and fairness.
However, we are not aware of detailed quantitative comparisons. The analysis for SD is very much simpler than for the Boston algorithms. In particular, the exit time is not interesting. In this section, we suppose that $n$ items and $n$ agents with Impartial Culture preferences are matched
by the Serial Dictatorship algorithm. 
\subsection{Groups of agents}
\label{ss:SD_group}

Results analogous to those in Sections \ref{s:naive_asymptotics} and
\ref{s:adaptive_asymptotics} are obtainable from the following explicit formula.

\begin{theorem}
\label{thm:SD_matching_probabilities}
The probability that the $k$th agent obtains his $s$th preference is $\binom{n-s}{k-s}\big/\binom{n}{k-1}$ for $s=1,\ldots,k$,
and zero for other values of $s$.
\end{theorem}

\begin{proof}{Proof.}
By the time agent $k$ gets an item, a random subset $T$ of $k-1$ of the $n$ items is already taken.
This agent's $s$th preference will be the best one left if and only if $T$ includes his first $s-1$ preferences,
but not the $s$th preference.
Of the $\binom{n}{k-1}$ equally-probable subsets $T$, the number satisfying this condition is $\binom{n-s}{k-s}$:
the remaining $k-s$ items in $T$ must be chosen from $n-s$ possibilities.
\hfill \qedsymbol
\end{proof}

In particular, the $n$th and last agent is equally likely to get each possible item.

\begin{corollary}[preference rank obtained]
\label{cor:SD_obtained}
Consider the serial dictatorship algorithm. Fix $s\geq1$ and a relative position $\theta\in[0,1]$.
The number $S_n(s,\theta)$ of members of $A_n(\theta)$ matched to their $s$th preference satisfies
$$ \frac1n S_n(s,\theta) \cvginprob \int_0^\theta q_s(\phi)\;d\phi
  $$
where $q_s(\theta) = \theta^{s-1}(1-\theta)$.
\end{corollary}

\begin{proof}{Proof.}
Let $p_{kn} = \binom{n-s}{k-s}\Big/\binom{n}{k-1}$.
Let $X_{kn}$ be the indicator of the event that the $k$th agent (of $n$) is matched to his $s$th
preference; thus $E[X_{kn}]=p_{kn}$ and $\hbox{Var}(X_{kn}) = p_{kn}(1-p_{kn})$.
The Impartial Culture model requires agents to choose their preferences
independently; thus the random variables $(X_{kn})_{k=1}^n$ are independent.
We have
$$ S_n(s,\theta)=\sum_{k=s}^{\lfloor n\theta\rfloor} X_{kn}
  $$
and so $E[S_n(s,\theta)] = \sum_{k=s}^{\lfloor n\theta\rfloor} p_{kn}$ and
$\hbox{Var}(S_n(s,\theta)) = \sum_{k=s}^{\lfloor n\theta\rfloor} p_{kn}(1-p_{kn})$.
Hence $\hbox{Var}(S_n(s,\theta)) \leq E[S_n(s,\theta)]$ and Lemma \ref{lem:C} applies.
It now remains only to show that $\frac1n E[S_n(s,\theta)] \to \int_0^\theta q_s(\phi)\;d\phi$.

Note that
$$ p_{kn} \;=\; \frac{(n-k+1)\cdot(k-1)(k-2)\cdots(k-s+1)}{n(n-1)\cdots(n-s+1)}
\;=\; \left(1-\frac{k-1}{n}\right)\prod_{j=1}^{s-1}\left(\frac{k-j}{n-j}\right) .
  $$
Hence,
$$ \frac1n E[S_n(s,\theta)]
\;=\; \frac1n \sum_{k=s}^{\lfloor n\theta\rfloor}
           \left(1-\frac{k-1}{n}\right)\prod_{j=1}^{s-1}\left(\frac{k-j}{n-j}\right)
\;=\; \int_0^\theta f_n(\phi)\;d\phi ,
  $$
where
$$ f_n(\phi) \;=\;
\begin{cases}
    \left(1-\frac{k-1}{n}\right)\prod_{j=1}^{s-1}\left(\frac{k-j}{n-j}\right)
    & \hbox{ for } \frac{k-1}{n}\leq\phi<\frac{k}{n}, \quad k=s,\ldots,\lfloor n\theta\rfloor\\
    0 & \hbox{ otherwise.}
\end{cases}
  $$
As $n\to\infty$, $f_n(\phi)\to(1-\phi)\phi^{s-1}$ pointwise;
since we also have $0\leq f_n(\phi)\leq1$, the dominated convergence theorem (\cite{Durrett})
ensures that $\int_0^\theta f_n(\phi)\;d\phi \to \int_0^\theta (1-\phi)\phi^{s-1}\;d\phi$.
\hfill \qedsymbol
\end{proof}

\subsection{Individual agents}
\label{ss:SD_indiv}

For individual agents, we have the following analogous result.

\begin{theorem}[preference rank obtained]
\label{thm:SD_individual_asymptotics}
Consider the serial dictatorship algorithm. Fix $s\geq1$ and a relative position $\theta\in[0,1]$.
The probability that agent $a_n(\theta)$ (the last with relative position at most $\theta$) 
is matched to his $s$th preference converges to $q_s(\theta) = \theta^{s-1}(1-\theta)$ as $n\to\infty$.
\end{theorem}

\begin{proof}{Proof.}
From Theorem~\ref{thm:SD_matching_probabilities}, this probability is
$$ \left(1-\frac{\lfloor n\theta\rfloor -1}{n}\right)\prod_{j=1}^{s-1}\left(\frac{\lfloor n\theta\rfloor-j}{n-j}\right) .
  $$
The result follows immediately.
\qedsymbol
\end{proof}

\section{Welfare}
\label{s:welfare}

In this section we obtain results on the utilitarian welfare achieved by the three mechanisms. We use the standard method of imputing utility to agents via scoring rules, since we know only their ordinal preferences.

\begin{defn}
\label{def:scoring_rule}
A \emph{positional scoring rule} is given by a sequence $(\sigma_n(s))_{s=1}^n$ of real numbers
with $0\leq\sigma_n(s)\leq\sigma_n(s-1)\leq1$ for $2\leq s\leq n$. 
\end{defn}

Commonly used scoring rules include \emph{$k$-approval} defined by $(1,1,\ldots,1,0,0,\ldots,0)$ where
the number of $1$'s is fixed at $k$ independent of $n$; when $k=1$ this is the usual plurality rule.
Note that $k$-approval is \textit{coherent}: for all $n$ the utility of a fixed rank object depends
only on the rank and not on $n$.
Another well-known rule is \emph{Borda} defined by $\sigma_n(s) = \frac{n-s}{n-1}$; Borda is not coherent.
Borda utility is often used in the literature, sometimes under the name ``linear utilities".

Each positional scoring rule defines an \emph{induced rank utility function}, common to all agents:
an agent matched to his $s$th preference derives utility $\sigma_n(s)$ therefrom.

Suppose (adopting the notation of Corollary \ref{cor:naive_obtained}, Corollary \ref{cor:adaptive_obtained},
and Corollary \ref{cor:SD_obtained})
that an assignment mechanism for $n$ agents matches $S_n(s,\theta)$ of the agents with relative position
at most $\theta$ to their $s$th preferences, for each $s=1,2,\ldots$.
According to the utility function induced by the scoring rule $(\sigma_n(s))_{s=1}^n$, the welfare
(total utility) of the agents with relative position at most $\theta$ is thus
\begin{equation}
\label{eq:welfare}
W_n(\theta) = \sum_{s=1}^n \sigma_n(s) S_n(s,\theta) .
\end{equation}

\begin{theorem}[Asymptotic welfare of the mechanisms]
\label{thm:welfare_asymptotics}
Assume an assignment mechanism with
$$ \frac1n S_n(s,\theta) \cvginprob \int_0^\theta q_s(\phi)\;d\phi 
\qquad\hbox{ as $n\to\infty$, for each $s=1,2,\ldots$}
  $$
where $\sum_{s=1}^\infty q_s(\theta)=1$.
Suppose the scoring rule $(\sigma_n(s))_{s=1}^n$ satisfies 
$$ \sigma_n(s) \to \lambda_s \qquad\hbox{ as $n\to\infty$, for each $s=1,2,\ldots$}
  $$
Then the welfare given by (\ref{eq:welfare}) satisfies
$$ \frac1n W_n(\theta) \cvginprob \sum_{s=1}^\infty \lambda_s \int_0^\theta q_s(\phi)\;d\phi .
  $$
\end{theorem}

\begin{proof}{Proof of Theorem \ref{thm:welfare_asymptotics}.}
For convenience, define $\sigma_n(s)=0$ when $n<s$; this allows us to write
$W_n(\theta)=\sum_{s=1}^{\infty}\sigma_n(s) S_n(s,\theta)$.
For any fixed $s'$, the finite sum $Y_n(s')$ defined by
$$ Y_n(s') = \sum_{s=1}^{s'} \left(\sigma_n(s)\frac{S_n(s,\theta)}{n} - \lambda_s \int_0^\theta q_s(\phi)\;d\phi\right)
  $$
has $Y_n(s')\cvginprob0$ as $n\to\infty$. We have
$$ \frac{W_n(\theta)}{n} - \sum_{s=1}^{\infty} \lambda_s \int_0^\theta q_s(\phi)\;d\phi
\;\;=\;\; Y_n(s')
\;+\; \sum_{s>s'} \sigma_n(s)\frac{S_n(s,\theta)}{n}
\;-\; \sum_{s>s'} \lambda_s \int_0^\theta q_s(\phi)\;d\phi
  $$
and so
\begin{equation}
\label{eq:welfare_cvg_bound}
\left|\frac{W_n(\theta)}{n} - \sum_{s=1}^{\infty} \lambda_s \int_0^\theta q_s(\phi)\;d\phi\right|
\;\;\leq\;\; \left|Y_n(s')\right|
\;+\; \sum_{s>s'} \frac{S_n(s,\theta)}{n}
\;+\; \sum_{s>s'} \int_0^\theta q_s(\phi)\;d\phi
\end{equation} 
(since $0\leq\sigma_n(s)\leq1$).
Note also that
$\sum_{s=1}^{s'}\frac{S_n(s,\theta)}{n}\;\cvginprob\;\sum_{s=1}^{s'}\int_0^\theta q_s(\phi)\;d\phi$,
while
$$ \sum_{s=1}^{\infty} \frac{S_n(s,\theta)}{n}
\;=\; \frac{\lfloor n\theta\rfloor}{n}
\;\to\; \theta
\;=\; \sum_{s=1}^{\infty} \int_0^\theta q_s(\phi)\;d\phi ,
  $$
and so
$$ \sum_{s>s'} \frac{S_n(s,\theta)}{n}
\;\cvginprob\; \sum_{s>s'} \int_0^\theta q_s(\phi)\;d\phi .
  $$
We can now establish the required convergence in probability.
Let $\epsilon>0$, and choose $s'$ so that $\sum_{s>s'} \int_0^\theta q_s(\phi)\;d\phi < \epsilon/3$.
Then (\ref{eq:welfare_cvg_bound}) gives
$$ P\left(\left|\frac{W_n(\theta)}{n} - \sum_{s=1}^{\infty} \lambda_s \int_0^\theta q_s(\phi)\;d\phi\right|>\epsilon\right)
\;\leq\; P\left(\left|Y_n(s')\right|>\epsilon/3\right)
\;+\; P\left(\sum_{s>s'} \frac{S_n(s,\theta)}{n}>\epsilon/3\right)
\;\to\; 0
  $$
as $n\to\infty$.
\hfill \qedsymbol
\end{proof}

Theorem \ref{thm:welfare_asymptotics} is applicable to naive Boston (via Corollary \ref{cor:naive_obtained}),
adaptive Boston (via Corollary \ref{cor:adaptive_obtained}), and serial dictatorship (via Corollary \ref{cor:SD_obtained}).

\begin{corollary}
\label{cor:welfare_asymptotics_k_approval}
The average $k$-approval welfare over all agents satisfies
$$ \frac1n W_n(1) \;\cvginprob\;
\begin{cases}
1 - \omega_{k+1} & \text{for Naive Boston}\\
(1 - e^{-1})  \sum_{\{(r,s): r\leq s\leq k\}} e^{1-r} u_{rs} & \text{for Adaptive Boston}\\
\frac{k}{k+1} & \text{for serial dictatorship}.
\end{cases}
$$
\end{corollary}

\begin{proof}{Proof.}
For the special case of $k$-approval utilities, the result of Theorem \ref{thm:welfare_asymptotics} reduces to
$$ \frac1n W_n(\theta) \cvginprob \sum_{s=1}^k \int_0^\theta q_s(\phi)\;d\phi .
  $$
Setting $\theta=1$ and using the expressions for $q_s(\phi)$ found in Corollary \ref{cor:naive_obtained}, Corollary
 \ref{cor:adaptive_obtained}, and Corollary \ref{cor:SD_obtained} yields the results.
\qedsymbol
\end{proof}

Corollary~\ref{cor:welfare_asymptotics_k_approval}  and Lemma~\ref{lem:wr_asymptotics} show that for each fixed $k$, Naive Boston has higher average welfare than Serial Dictatorship. This is expected, because Naive Boston maximizes the number of agents receiving their first choice, then the number receiving their second choice, etc. Adaptive Boston apparently scores better than Serial Dictatorship for each $k$, although we do not have a formal proof. Figure~\ref{fig:welf comp} illustrates this for $1\leq k \leq 10$. Already for $k=3$, where the limiting values are 0.75, 0.776 and 0.803, the algorithms give similar welfare results, and they each asymptotically approach $1$ as $k\to \infty$.

\begin{table}[hbtp]
\centering
\begin{tabular}{|c|c|c|c|}
\hline
algorithm & $k=1$ & $k=2$ & $k=3$\\
\hline
Naive Boston & $1-e^{-1} \approx 0.632$  & $1 - e^{-1}e^{-e^{-1}}\approx 0.745$ & $1 - e^{-1}e^{-e^{-1}}e^{-e^{-e^{-1}}}\approx 0.803$\\
Adaptive Boston  & $1-e^{-1} \approx 0.632$ & $(1-e^{-1})(1+e^{-2}) \approx 0.718$ & $(1-e^{-1})(1+2e^{-2}-e^{-3}+e^{-5})\approx 0.776$\\
Serial Dictatorship & $1/2 = 0.500$ & $2/3\approx 0.667$ & $3/4=0.750$ \\
\hline
\end{tabular}
\vspace{10pt}
\caption{Limiting values as $n\to \infty$ of $k$-approval welfare.}
\label{t:app_limits}
\end{table}

\begin{figure}
    \centering
    \includegraphics[width=0.8\textwidth]{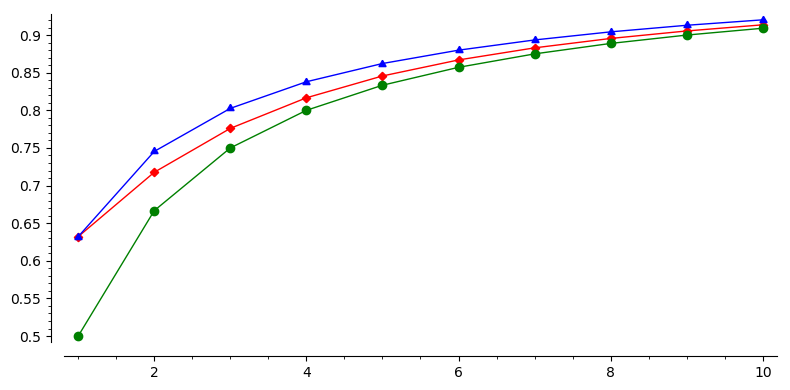}
    \caption{Limiting values as $n\to \infty$ of $k$-approval welfare for $1\leq k \leq 10$. Top: Naive Boston. Middle: Adaptive Boston. Bottom: Serial Dictatorship.}
    \label{fig:welf comp}
\end{figure}

\begin{corollary}
\label{cor:welfare_asymptotics_Borda}
For an assignment mechanism as in Theorem \ref{thm:welfare_asymptotics},
the Borda welfare satisfies
$$ \frac1n W_n(\theta) \cvginprob \theta .
  $$
\end{corollary}

\begin{corollary}
\label{cor:welfare_asymptotics_compare_borda}
For each of Naive Boston, Adaptive Boston and Serial Dictatorship, the average normalized Borda welfare over all agents is asymptotically equal to $1$.
\end{corollary}

\begin{remark}
Note that the Borda utility of a fixed preference rank $s$ has the limit $\lambda_s=1$, meaning that,
in the asymptotic limit as $n\to\infty$, agents value the $s$th preference (of $n$) just as highly as the first preference.
Consequently, mechanisms such as serial dictatorship or the Boston algorithms, which under IC are able to
give most agents one of their first few preferences, achieve the same asymptotic Borda welfare as if
every agent were matched to his first preference. This behaviour is really a consequence of the normalization of the
Borda utilities $\sigma_n(s) = \frac{n-s}{n-1}$ to the interval $[0,1]$: the first few preferences all have
utility close to 1.
\end{remark}

\section{Order bias}
\label{s:order_bias}

A recently introduced \cite{FrPW2021} average-case measure of fairness of discrete allocation algorithms is \emph{order bias}. The relevant definitions are recalled here for an arbitrary discrete assignment algorithm $\algo$ that fixes an order on agents (such as the order $\rho$ assumed in the present paper).

\begin{defn}
The \emph{expected rank distribution} under $\algo$ is the mapping $D_\algo$
on $\{1, \dots, n\} \times \{1, \dots, n\}$ whose value at $(r,j)$ is the probability under IC
that $\algo$ assigns the $r$th agent his $j$th most-preferred item.
\end{defn}
We usually represent this mapping as a matrix where the rows represent agents and the columns represent items. 

\begin{defn}
Let $u$ be a common rank utility function for all agents: $u(j)$ is the utility derived by an agent
who obtains his $j$th preference. Define the order bias of $\algo$ by
$$ \beta_n(\algo; u) = \frac{\max_{1\leq p, q \leq n} |U(p) - U(q)|}{u(1)-u(n)},
  $$
where $U(p)=\sum_{j=1}^n D_\algo(p,j) u(j)$, the expected utility of the item obtained by
the $p$th agent.
\end{defn}

It is desirable that $\beta_n$ be as small as possible, out of fairness to each position in the order in the absence of any knowledge of the profile.

The mechanisms in this paper (naive and adaptive Boston, and serial dictatorship) treat agents
unequally by using a choosing/tiebreak order $\rho$. In all of these mechanisms, the first agent in
$\rho$ always obtains his first-choice item, and so has the best possible expected utility.
The last agent in $\rho$ has the smallest expected utility; this is a consequence of the
following result.

\begin{theorem}[Earlier positions do better on average]
\label{thm:last_agent_is_worst_off}
Let $a$ be an agent in an instance of the house allocation problem with IC preferences.
Let the random variable $S$ be the preference rank of the item obtained by $a$.
The naive and adaptive Boston mechanisms and serial dictatorship all have the property that
for all $s\geq1$, $P(S>s)$ is monotone increasing in the relative position of $a$
({\it i.e.} greater for later agents in $\rho$).
\end{theorem}

\begin{remark}
Thus in the expected rank distribution matrix, each row stochastically dominates the one below it.
For each common rank utility function $u$, the expected utility of agent $a$ is $u(1) + \sum_{s=1}^{n-1} (u(s+1)-u(s))\;P(S>s)$,
so Theorem \ref{thm:last_agent_is_worst_off} implies that the expected utility is
monotone decreasing in the relative position of $a$. In particular, the 
first agent has the highest and the last agent the lowest expected utility.

\end{remark}

\begin{proof}{Proof of Theorem \ref{thm:last_agent_is_worst_off}.}
Let $a_1$ and $a_2$ be consecutive agents, with $a_2$ immediately after $a_1$ in $\rho$.
Let $S_1$ and $S_2$ be the preference ranks of the items obtained by $a_1$ and $a_2$.
It will suffice to show that $P(S_1>s)\leq P(S_2>s)$.
To this end, consider an alternative instance of the problem in which $a_1$ and $a_2$
exchange preference orders before the allocation mechanism is applied.
We will refer to this instance and the original one as the ``exchanged'' and ``non-exchanged''
processes respectively.
Denote by $S'_1$ and $S'_2$ the preference ranks of the items obtained by $a_1$ and $a_2$
in the exchanged process.
Since the exchanged process also has IC preferences, $S_1$ and $S'_1$ have the same probability
distribution; similarly $S_2$ and $S'_2$.

We now show that all three of our allocation mechanisms have the property that $S_1\leq S'_2$.
From this the result will follow, since $S_1\leq S'_2\implies P(S_1>s)\leq P(S'_2>s) = P(S_2>s)$.

For serial dictatorship, the exchanged and non-exchanged processes evolve identically for
agents preceding $a_1$ and $a_2$. In the non-exchanged process, agent $a_1$ then finds that his
first $S_1-1$ preferences are already taken; in the exchanged process, these same items are the
first $S_1-1$ preferences of $a_2$. Hence, $S'_2\geq S_1$.

For the Boston mechanisms, let $R$ be the number of unsuccessful bids made by $a_1$ in the
non-exchanged process. Then the exchanged and non-exchanged processes evolve identically for
the first $R$ rounds, except that the bids of $a_1$ and $a_2$ are made in reversed order; this
reversal has no effect on the availability of items to other agents. After these $R$ rounds,
$a_1$ (in the non-exchanged process) and $a_2$ (in the exchanged process) have reached the
same point in their common preference order; in the next round both will bid for the $S_1$th
preference in this order. Hence, $S'_2\geq S_1$.
\hfill \qedsymbol
\end{proof}

The order bias of Serial Dictatorship is easy to analyse.

\begin{theorem}
\label{thm:sd_bias_app_exact}
Fix $k\geq 1$ and $n\geq 1$. Then
\begin{enumerate}[(i)]
\item The $k$-approval order bias for Serial Dictatorship equals $1 - \frac{k}{n}$.
\item The Borda order bias for Serial Dictatorship equals $1/2$.
\end{enumerate}
\end{theorem}
\begin{proof}{Proof.}
The probability of getting each choice is $1/n$ for the last agent. Hence the expected utility under $k$-approval for that agent is $k/n$. The first agent always gets its first choice. This yields (i). For (ii), note that for the last agent, the probability of getting each rank in his preference order is $1/n$. Hence the expected utility under Borda for that agent is 
$$\frac{1}{n} \sum_{j=1}^{n} \frac{n-j}{n-1} = \frac{1}{n(n-1)} \sum_{j=0}^{n-1} j 
= \frac{1}{2}.$$
Again, the first agent always gets his first choice.
\hfill \qedsymbol
\end{proof}

\begin{corollary}
\label{cor:sd_bias_app_asymp}
For each fixed $k$, the $k$-approval order bias of SD is asymptotically equal to $1$ and the Borda order bias is asymptotically equal to $1/2$.
\end{corollary}

We now move to the Boston mechanisms.     

\begin{theorem}
\label{thm:na_boston bias}

For each fixed $k$, the $k$-approval order bias of Naive Boston is asymptotically
$z'_{k+1}(1).$ 
\end{theorem}
\begin{proof}{Proof.}
Since the first agent always gets its top choice with utility $1$, it follows that
$\beta_n(NB)$ equals the probability that the last agent survives until round $k+1$, which asymptotically equals $z'_{k+1}(1)$.
\qedsymbol
\end{proof}

\begin{theorem}
\label{thm:adaptive_kapproval_bias}
For each fixed $k$, the $k$-approval order bias of Adaptive Boston is asymptotically
$$
1 - e^{-1} \sum_{\{(r,s): r\leq s \leq k\}}\left(1 - e^{-1}\right)^{r-1}  u_{rs}.
  $$
\end{theorem}

\begin{proof}{Proof.}
The probability that the last agent in $\rho$ is matched to one of his first $k$
preferences is $\sum_{s=1}^k D_\algo(n,s)$. According to Theorem
\ref{thm:adaptive_individual_asymptotics} Part \ref{adaptive_individual_obtained},
the asymptotic limit of this quantity is $\sum_{s=1}^k q_s(1)$, where
$$ q_s(1) \;=\; \sum_{r=1}^s u_{rs}(y'_r(1) - y'_{r+1}(1)) .
  $$
The asymptotic order bias is thus
$$ \lim_n \left(1 - \sum_{s=1}^k D_\algo(n,s)\right)
\;=\; 1 - \sum_{s=1}^k \sum_{r=1}^s u_{rs}(y'_r(1) - y'_{r+1}(1)) .
  $$
As noted in Remark \ref{rem:adaptive_yr}, we have
$y'_r(1)=(1-e^{-1})^{r-1}$. The result follows.
\qedsymbol
\end{proof}

\begin{theorem}
\label{thm:boston_borda_bias}
The Borda order bias of each Boston mechanism is asymptotically zero.
\end{theorem}
\begin{proof}{Proof.}
Let $\ell_n$ denote the expected Borda utility of the last agent in $\rho$,
that is
$$ \ell_n \;=\; \sum_{s=1}^n \left(\frac{n-s}{n-1}\right)D_\algo(n,s) .
  $$
Then for any $s_0$,
$$ \liminf_n \; \ell_n
\;\geq\; \liminf_n \; \left(\frac{n-s_0}{n-1}\right) \sum_{s=1}^{s_0} D_\algo(n,s)
\;=\; \sum_{s=1}^{s_0} q_s(1) ,
  $$
where $q_s(1) = \lim_n D_\algo(n,s)$, as given by Theorem \ref{thm:naive_individual_asymptotics} (naive Boston) and Theorem
\ref{thm:adaptive_individual_asymptotics} (adaptive Boston).
Since $\sum_{s=1}^{\infty}q_s(1)=1$ (see Remarks \ref{rem:naive_qs} and
\ref{rem:adaptive_sum_qs}) and $s_0$ was arbitrary, we obtain
$\lim_n \ell_n = 1$. The order bias is $1-\ell_n$; hence the result.
\qedsymbol
\end{proof}

\begin{table}[hbtp]
\centering
\begin{tabular}{|l|l|l|l|}
\hline
algorithm & $k=1$ & $k=2$ & $k=3$\\
\hline
NB & $1-e^{-1}\approx 0.632$ & $(1 - e^{-1})(1-e^{-1}e^{-e^{-1}}) \approx 0.471$
   & $(1 - e^{-1})(1 - e^{-1}e^{-e^{-1}})(1 - e^{-1}e^{-e^{-1}}e^{-e^{-{e^{-1}}}})\approx 0.378$ \\
AB & $1-e^{-1}\approx 0.632$ & $(1-e^{-1})(1-e^{-2})\approx 0.547$ & $(1-e^{-1})(1-e^{-2}) - (1-e^{-1})^2(e^{-2}+e^{-4}) \approx 0.485$ \\
SD & 1 & 1 & 1\\
\hline
\end{tabular}
\vspace{10pt}
\caption{Limiting quantities for $k$-approval order bias.}
\end{table}

\begin{figure}
    \centering
    \includegraphics[width=0.8\textwidth]{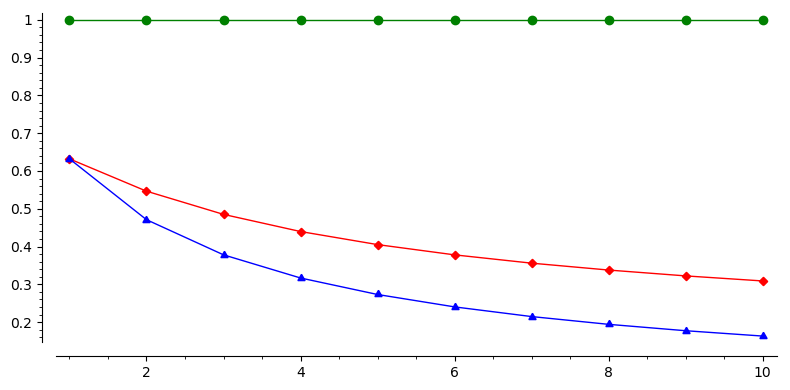}
    \caption{Limiting values as $n\to \infty$ of $k$-approval order bias for $1\leq k \leq 10$. Top: Serial Dictatorship. Middle: Adaptive Boston. Bottom: Naive Boston.}
    \label{fig:bias comp}
\end{figure}

%
%
%

\section{Conclusion}
\label{s:conclude}

If we relax the IC assumption on preferences, we should expect different results, although the relative performance of the three algorithms will likely not vary. For example, simulations \cite{FrPW2021} with preferences drawn from the Mallows distribution show that for small values of the Mallows dispersion parameter it is much harder to satisfy all agents or keep order bias low, but nevertheless NB beats AB, which beats SD, over the entire range of parameters.

A striking feature of our results, under the IC assumption on preferences and assuming sincere agent behavior, is that although the Boston algorithms have a welfare advantage over Serial Dictatorship, the advantage is rather small. 

The limiting results for average welfare gained by the agents up to position $\theta$ in the choosing order show that the limit is concave in $\theta$. For the Boston mechanisms, this concavity is slight: for example, even for plurality utilities the median of the cumulative Adaptive Boston welfare distribution occurs at position approximately $0.378$, and this becomes even more evenly distributed as $k$ increases and we choose $k$-approval utilities (the limiting case is the same as Borda, where the cumulative distribution is linear).  

However, there is a huge difference in the values of the more egalitarian fairness criterion order bias, with SD being asymptotically as biased as it could be, and the Boston algorithms being asymptotically unbiased with respect to our normalized Borda utilities and having much lower bias than SD even for utilities such as $k$-approval for small $k$. 

Thus Naive Boston beats Adaptive Boston on both welfare and order bias, and Adaptive Boston beats Serial Dictatorship. From a welfare viewpoint, then, SD should be avoided. Of course, there are always tradeoffs. A persistent theme of the research literature is the inevitable tradeoff between strategyproofness, economic efficiency and agent welfare, and there is still much to be learned about these issues. SD is strategyproof, while AB gives less incentive to strategize than NB \cite{MeSe2021}.

The order bias of the Boston algorithms, although smaller than that of SD, is still rather large. Thus if this fairness criterion is important, it makes sense to use a mechanism like Top Trading Cycles, which is strategyproof and has zero order bias in this situation \cite{FrPW2021}. Note that since TTC (with a randomly chosen endowment) is equivalent to SD \cite{AbSo1998}, and SD does not give up much in welfare to NB, TTC may be a good choice if preferences of agents are well described by IC. 

A simple idea that will reduce order bias is to reverse the order in which agents choose at each round (or just at the second round). Quantifying the improvement via an analysis analogous to that in this paper is not easy, because 
it is no longer clear that the worst off agent will be the initially last one in the choosing order. We leave this for future work.

The Boston algorithms discussed here are specializations of algorithms used for school choice to the case where each school has a single seat and schools have a common preference order over applicants. 
Further analysis of school choice mechanisms in the general case, from the viewpoint of welfare and order bias, would be very desirable.

We have studied only sincere behavior by agents. Strategic behavior under the Boston mechanisms does occur in practice, and does cause welfare loss, but the social welfare cost of adopting a strategyproof alternative such as (random) Serial Dictatorship is often substantial, as shown in analysis of Harvard course matching \cite{BuCa2012}. It would be interesting to explore this issue further in the housing allocation model, and to study welfare and 
order bias in the multi-unit assignment model used in \cite{BuCa2012}.
 
The $k$-approval utilities we have used here are widely used in assignment applications.
For example, statistics such as the fraction of school choice students obtaining one of their top three choices,
or their one favorite course, are commonly discussed.



\bibliographystyle{plain}
\bibliography{assignment.bib} 


\end{document}